\newif\ifconf\conffalse
\ifconf
\documentclass[letterpaper]{sig-alternate}
\else
\documentclass[letterpaper,11pt]{article}
\fi
\usepackage{url}
\ifconf
\else
\usepackage{amsthm}
\usepackage{fullpage}
\fi
\usepackage{amssymb,amsmath,amsfonts,comment,cite}
\usepackage{xspace}
\usepackage[usenames]{color}
\usepackage{hyperref}

\DeclareMathOperator{\polylog}{polylog}

\DeclareMathSymbol{\qedsymb} {\mathord}{AMSa}{"04}

\newcommand{\eps}{\varepsilon}

\newcommand{\abs}[1]{\left\lvert #1 \right\rvert}

\newcommand{\setst}[2]{\left\{\; #1 \,:\, #2 \;\right\}}        

\newcommand{\prob}[1]{\mathbb{P}\left[\,#1\,\right]}
  
\newcommand{\oct}{\quad\quad}                                   

\newcommand{\norm}[1]{\left\lVert #1 \right\rVert}
\newcommand{\gennorm}[2]{\left\lVert #2 \right\rVert_{#1}}

\newcommand{\N}{\ensuremath{\mathcal{N}}}
\newcommand{\ip}[2]{\ensuremath{\left\langle #1,#2\right\rangle}}

\newcommand{\R}{\mathbb{R}}
\newcommand{\C}{\mathbb{C}}

\DeclareMathOperator*{\E}{\mathbb{E}}

\renewcommand{\Pr}{\mathbb{P}\xspace}

\newcommand{\EquationName}[1]{\label{eq:#1}}

\newcommand{\LemmaName}[1]{\label{lem:#1}}

\newcommand{\SectionName}[1]{\label{sec:#1}}

\newcommand{\FigureName}[1]{\label{fig:#1}}

\newcommand{\Equation}[1]{Eq.\:\eqref{eq:#1}}

\newcommand{\Lemma}[1]{Lemma~\ref{lem:#1}}

\newcommand{\Figure}[1]{Figure~\ref{fig:#1}}

\newtheorem{theorem}{Theorem}

\newtheorem{definition}[theorem]{Definition}

\newtheorem{lemma}[theorem]{Lemma}

\newtheorem{prop}[theorem]{Proposition}

\newcommand{\proofbelow}{3pt}
\newcommand{\afterproof}{\hfill $\blacksquare$ \par \vspace{\proofbelow}}
\newcommand{\aftersubproof}{\hfill $\Box$ \par \vspace{\proofbelow}}
\renewenvironment{proof}{\noindent\textbf{Proof.}\,}{\afterproof}

\renewcommand{\th}{\ifmmode{^{\textrm{th}}}\else{\textsuperscript{th}\ }\fi}

\newcommand{\G}{\mathcal{G}}
\DeclareMathOperator{\supp}{supp}

\newcommand{\decnorm}[2]{\left\|#1\right\|_{#2}}
\newcommand{\fronorm}[1]{\decnorm{#1}{F}}
\newcommand{\opnorm}[1]{\decnorm{#1}{}}
\newcommand{\twonorm}[1]{\decnorm{#1}{2}}
\newcommand{\Xnorm}[1]{\decnorm{#1}{X}}
\newcommand{\inftynorm}[1]{\decnorm{#1}{\infty}}

\begin{document}

\ifconf
\else
\author{Jelani Nelson\thanks{Institute for Advanced Study. \texttt{minilek@ias.edu}. Supported by NSF
  CCF-0832797 and NSF DMS-1128155.}\oct
Eric Price\thanks{MIT. \texttt{ecprice@mit.edu}. Work supported by an NSF Graduate Research Fellowship and a Simons Fellowship}\oct
Mary Wootters\thanks{University of Michigan, Ann Arbor. \texttt{wootters@umich.edu}. Supported by NSF CCF-0743372 and NSF CCF-1161233.}
}

\fi

\title{New constructions of RIP matrices\\ with fast multiplication and fewer rows}

\maketitle

\begin{abstract}
In compressed sensing, the {\it restricted isometry property} (RIP)
is a sufficient condition for the efficient reconstruction of a nearly $k$-sparse
vector $x \in \C^d$ from $m$ linear measurements $\Phi x$.  It is 
desirable for $m$ to be small, and for $\Phi$ to support fast matrix-vector multiplication.  
In this work, we give a randomized construction of RIP matrices $\Phi\in\C^{m\times
  d}$, preserving the $\ell_2$ norms of all $k$-sparse vectors with
distortion $1+\eps$,
where the matrix-vector multiply $\Phi x$
can be computed in nearly linear time.  The number of rows $m$
is on the order of $\eps^{-2}k\log d\log^2(k\log d)$. 
 Previous analyses of constructions of RIP matrices supporting
fast matrix-vector multiplies, such as the sampled
discrete Fourier matrix, required $m$ to be larger by roughly
a $\log k$ factor. 

Supporting fast matrix-vector
multiplication is useful for iterative recovery algorithms which
repeatedly multiply by $\Phi$ or $\Phi^*$.
Furthermore, our construction, together with a connection
between RIP matrices and the Johnson-Lindenstrauss lemma  in
[Krahmer-Ward, SIAM.\ J.\ Math.\ Anal.\ 2011], implies fast
Johnson-Lindenstrauss embeddings with asymptotically fewer rows than
previously known. 

Our approach is a simple twist on previous constructions.
Rather than choosing the rows for the embedding matrix to be
rows sampled from some larger structured matrix (such as the discrete
Fourier transform or a random circulant matrix), we instead choose
each row of the embedding
matrix to be a linear combination of a small number of rows of the original matrix,
with random sign flips as coefficients.
The main tool in our analysis is a recent bound for the
supremum of certain types of Rademacher chaos processes in
[Krahmer-Mendelson-Rauhut, arXiv abs/1207.0235].
\end{abstract}


\section{Introduction}\SectionName{intro}
The goal of {\it compressed sensing}~\cite{crt06b,don06} is to
efficiently reconstruct sparse, high-dimensional signals from a small
set of linear measurements.  We say that a $x \in \C^d$ is {\it
  $k$-sparse } if $\|x\|_0 \leq k$, where $\|x\|_0$ denotes the number
of non-zero entries.  The idea is that if $x$ is guaranteed to be
sparse or nearly sparse (that is, close to a sparse vector), then we
should be able to recover it with far fewer than $d$ measurements.
Organizing the measurements as the rows of a matrix $\Phi \in \C^{m
  \times d}$, one wants an efficient algorithm $\mathcal{R}$ which
approximately recovers a signal $x \in \C^d$ from the measurements
$\Phi x$; that is, $\|\mathcal{R}(\Phi x) - x\|_2$ should be small.
There are several goals in the design of $\Phi$ and $\mathcal{R}$.  We
would like $m \ll d$ to be as small as possible, so that $\Phi x$ can
be interpreted as a compression of $x$.  We also ask that the recovery
algorithm $\mathcal{R}$ be efficient, and satisfy a reasonable
recovery guarantee when $x$ is close to a sparse vector.


The recovery guarantee most popular in the literature is the
\emph{$\ell_2/\ell_1$ guarantee}, which compares the error between $x$
and the recovery $\mathcal{R}(\Phi x)$ to the error between $x$ and
the best $k$-sparse approximation of $x$.  More precisely, to satisfy
the $\ell_2/\ell_1$ guarantee there must exist a constant $C$ such
that for every $x$, $\mathcal{R}(\Phi x)$ satisfies
\begin{align}\label{e:l2l1}
  \|\mathcal{R}(\Phi x) - x\|_2 \le \frac{C}{\sqrt{k}}\cdot
  \inf_{\substack{y\in\C^d\\ \|y\|_0 \leq k}} \|x - y\|_1.
\end{align}
The value of $m$ and the pair $\Phi, \mathcal{R}$ can depend on $d$
and $k$.  Above, $\|\cdot \|_p$ denotes the $\ell_p$ norm $\|x\|_p =
\left(\sum_i |x_i|^p\right)^{1/p}$ and $\norm{x}_0$ denotes the
number of non-zero entries of $x$.

In this work, we will be concerned with a sufficient condition for the
$\ell_2/\ell_1$ guarantee, known as the {\it $(\eps, 2k)$ restricted
  isometry property}, or {\it $(\eps, 2k)$-RIP}.  We say that a matrix
$\Phi\in\C^{m\times d}$ has the $(\eps,k)$-RIP if
\begin{align}\label{e:RIP}
  \forall x\in\C^d ,\ \|x\|_0 \le k \Rightarrow (1-\eps)\|x\|_2^2\le
  \|\Phi x\|_2^2\le (1+\eps)\|x\|_2^2 .
\end{align}
It is known that if $\Phi$ satisfies the $(\eps,k)$-RIP for $\eps <
\sqrt{2} - 1$, then $\Phi$ enables the $\ell_2/\ell_1$ guarantee for
some constant $C$ \cite{CT05,crt06,Candes08}.  Furthermore,
this guarantee is achievable by efficient methods such as solving a linear
program~\cite{crt06,CSD01,DET06}.

In this work, we construct matrices $\Phi$
which satisfy the RIP with few rows, and which additionally support
fast matrix-vector multiplication.  The speed of the encoding time is
important not just for encoding $x$ as $\Phi x$, but also for the
reconstruction of $x$.  Aside from linear programming, there are
several iterative algorithms for recovering $x$ from $\Phi x$ when
$\Phi$ satisfies the RIP: for example Iterative Hard Thresholding
\cite{BD08}, Gradient Descent with Sparsification \cite{GK09}, CoSaMP
\cite{NT09}, Hard Thresholding Pursuit \cite{Foucart11}, Orthogonal
Matching Pursuit \cite{TG07}, Stagewise OMP (StOMP) \cite{DTDS12}, and
Regularized OMP (ROMP) \cite{NV09,NV10}.  All these algorithms have
running times essentially bounded by the number of iterations (which
is usually logarithmic in $d$ and an error parameter) times the
running time required to perform a
matrix-vector multiply with either $\Phi$ or $\Phi^*$, and so it is
important that this operation be fast.

If we do not require fast matrix-vector multiplication, it is known that RIP matrices  
exist with
 $m = \Theta(k\log (d/k))$.
For example, any matrix with i.i.d.\ Gaussian or subgaussian entries suffices~\cite{CT06,PMT:2007,BDDW08}.
This is known to be optimal even
for the $\ell_2/\ell_1$ recovery problem itself via a connection to
Gelfand widths \cite{Kashin77,GG84}
(see a discussion in \cite[Section 3]{BDDW08}), and is even required
to obtain a weaker randomized guarantee \cite{DIPW10}. 
However, for such matrices, na\"{i}ve matrix-vector multiplication requires time $O(dm)$.  Ideally, for the applications above,
this would instead be nearly linear in $d$. 
This has caused a search for RIP matrices that support fast matrix-vector
multiplication, leading to constructions that unfortunately require $m$ to be larger than the optimal by several
logarithmic factors.
We discuss previous work in closing this gap, and our contribution, in more detail in Section \ref{ssec:previouswork} below.

\subsection{Johnson-Lindenstrauss}

The Johnson-Lindenstrauss (JL) lemma of \cite{JL84} is related to the RIP, and,
as we will see below, our constructions of RIP matrices will imply constructions
of Johnson-Lindenstrauss transforms with fast embedding time.
The JL lemma states that there is a way to embed $N$ points in $\ell_2^d$ into a linear subspace
of dimension approximately $\log N$, with very little distortion.
\footnote{The JL lemma is most commonly stated over $\R$, so we state it this way here.  However, 
as in~\cite{KW11},
all of our results extend to complex vectors and complex matrices.}
\begin{lemma}\LemmaName{jl-lemma}
For any $0<\eps<1/2$ and any $x_1,\ldots,x_N\in \R^d$, there exists a
linear map $A\in \R^{m\times d}$ for $m = O(\eps^{-2}\log N)$ such
that for all $1 \le 1 < j \le N$,
$$ (1-\eps)\|x_i - x_j\|_2 \le \|Ax_i - Ax_j\|_2
\le (1+\eps)\|x_i - x_j\|_2 .$$
\end{lemma}

For any fixed set of vectors $x_1, \ldots, x_N$, we call a matrix
$A$ as in the lemma an {\it $\eps$-JL matrix} for that set.
It is known that there are sets of $N$ vectors for which 
$m = \Omega((\eps^{-2}/\log(1/\eps))\log N)$ is
required~\cite{Alon03}. In fact, this 
bound holds for any, not necessarily linear, embedding into $\ell_2^m$.

The JL lemma is a useful tool for speeding up solutions to several
problems in high-dimensional computational geometry; see for example
\cite{Indyk01, Vempala04}. Often, one has an
algorithm which is fast in terms of the number
of points but slow as a function of dimension: a good strategy to approximate a solution
quickly is to first reduce the input
dimension via the JL lemma before running the algorithm.
Recently dimensionality reduction via linear maps has also found
applications in approximate
numerical algebra problems such as linear regression and low-rank
approximation \cite{Sarlos06,CW09,CW12,MM12,NN12a}, and for the $k$-means
clustering problem \cite{BZMD11}. Going back to our original problem, the JL lemma
also implies
the existence of $(\eps,k)$-RIP matrices with $O(\eps^{-2}k\log(d/k))$
rows \cite{BDDW08}.

Due to its algorithmic importance, it is of interest to obtain JL matrices
which allow for fast embedding time, i.e.\ for which the matrix-vector product $Ax$ can
be computed quickly.
Paralleling the situation with the RIP, if we do not require that $A$ support fast matrix-vector multiplication,
there are many constructions of dense matrices $A$ which are JL matrices with high probability 
~\cite{Achlioptas03,AV06,DG03,FM88,IM98,JL84,Matousek08}.
For example, we may take $A$ to have i.i.d.\ Gaussian or subgaussian entries.
However, for such $A$ matrix-vector multiplication takes time $O(dm)$, where as before we would like it to be nearly linear in $d$. 
As with the RIP, if we require this embedding time, there is gap of several logarithmic factors between the upper and lower bounds on the target dimension $m$.  We
review previous work and state our contributions on this gap below.

\subsection{Previous Work on Fast RIP/JL, and Our Contribution}
\label{ssec:previouswork}
Above, we saw the importance of constructing RIP and JL matrices which 
not only have few rows but also support fast matrix-vector multiplication.
Below, we review previous work in this direction.  We then state our
contributions and improvements, which are summarized in \Figure{results}.

The best known construction of RIP matrices with fast multiplication 
come from either subsampled Fourier matrices (or related
constructions) or from partial circulant matrices.
Cand{\`e}s and Tao showed in~\cite{CT06} that a matrix whose rows are 
$m = O(k\log^6 d)$ random rows from the Fourier matrix satisfies the
$(O(1), k)$-RIP with positive probability.  The analysis of Rudelson
and Vershynin~\cite{RV08} and an optimization of it by
Cheraghchi, Guruswami, and Velingker~\cite{CGV13} improved the number
of rows required for the $(\eps,k)$-RIP to $m = O(\eps^{-2} k\log d
\log^3 k )$.
For circulant matrices, initial works
required $m \gg k^{3/2}$ to obtain the $(\eps,k)$-RIP
\cite{HBRN10,RRT12}; Krahmer, Mendelson and Rauhut~\cite{KMR12}
recently improved
the number of rows required to $m = O(\eps^{-2}k\log^2 d\log^2 k)$.

The first work on JL matrices with fast multiplication was by
Ailon and Chazelle~\cite{AC09}, which had $m = O(\eps^{-2} \log N)$ rows
and embedding time $O(d\log d + m^3)$.
In certain applications $N$ can be exponentially
large in a parameter of interest, e.g.\
when one wants to preserve the geometry of an entire subspace
for numerical linear algebra \cite{CW09,Sarlos06} or $k$-means
clustering \cite{BZMD11}, or the set of all sparse vectors in compressed
sensing \cite{BDDW08}. Thus, while the number of rows in this construction is optimal,
for some applications it is important to
improve the dependence on $m$ in the running time.
Ailon and Liberty~\cite{AL09} improved the running time to $O(d\log m +
m^{2+\gamma})$ for any desired $\gamma>0$ (with the same number of rows), 
and more recently the same authors gave a construction with $m = O(\eps^{-4}\log
N\log^4 d)$ supporting matrix-vector multiplies in time $O(d\log d)$
\cite{AL11}. Krahmer and Ward~\cite{KW11} improved the target
dimension to $m = O(\eps^{-2}\log
N\log^4 d)$. 

This last improvement of \cite{KW11} is actually a more general
result.  Specifically, they showed that, when the columns are
multiplied by independent random signs, any $(O(\eps),O(\log N))$-RIP
matrix becomes an $\eps$-JL matrix for a fixed set of $N$ vectors with
probability $1 - N^{-\Omega(1)}$. Since we saw above that sampling
$O(\eps^{-2}k\log d\log^3k)$ rows from the discrete Fourier or Hadamard matrix
satisfies $(\eps,k)$-RIP with constant probability, conditioning on
this event and applying the result of \cite{KW11} implies a JL matrix
with $m = O(\eps^{-2}\log N\log d\log^3(\log N)) = O(\eps^{-2}\log
N\log^4 d)$ and embedding time $O(d\log d)$.  We will use the same
method to obtain fast JL matrices from our constructions of RIP
matrices.

Another way to obtain JL matrices which support fast matrix-vector multiplication
is to construct sparse JL matrices~\cite{WDLSA09,DKS10,KN10,BOR10,KN12}.  
These constructions allow for very fast multiplication $Ax$ when the vector $x$ is itself sparse.
However, 
these constructions have an $\Omega(\eps)$ fraction of nonzero
entries, and it is known that any JL transform with $O(\eps^{-2}\log
N)$ rows requires an $\Omega(\eps/\log(1/\eps))$ fraction of nonzero
entries \cite{NN12b}.  Thus, for constant
$\eps$ and dense $x$, multiplication still requires time $\Theta(dm)$.

In this work we propose and analyze a new method for constructing RIP
matrices that support fast matrix-vector multiplication.  Loosely
speaking, our method takes any ``good'' ensemble of RIP matrices, and
produces an ensemble of RIP matrices with fewer rows by multiplying by
a suitable hash matrix.  We can apply our method to either subsampled
Fourier matrices or partial circulant matrices to obtain our improved
RIP matrices. 

Our construction follows a natural
intuition. For example, let $A$ be the discrete Fourier matrix, and suppose that $S$
is an $m \times d$ matrix with i.i.d.\ Rademacher entries, appropriately normalized. 
If $m = \Theta(\eps^{-2}k\log(d/k))$, then $SA$ satisfies the $(\eps,k)$-RIP with high probability,
because $S$ has the RIP, and $A$ is an isometry.
Unfortunately, this construction has slow matrix-vector
multiplication time.  On the other hand, if $S'$ is an extremely sparse
random sign matrix, with only one non-zero per row, then $S'A$ is a subsampled Fourier matrix,
supporting fast multiplication.  Unfortunately, in order to show that $S'A$ satisfies the RIP with high probability, 
$m$ must be increased by $\polylog(k)$ factors.
This raises the question: can we get the best of both worlds?
How sparse must the sign matrix $S$ be to ensure RIP with few rows, and can it be sparse enough 
to maintain fast matrix-vector multiplication?
In some sense, this question, and our results, connects the two lines of research---structured matrices and sparse matrices---on fast JL matrices mentioned above.  
Our results imply we can improve the number of rows over
previous work by using such a sparse sign matrix with only
$\polylog(d)$ non-zeroes per row. 

\paragraph{Our Main Contribution:} We give randomized constructions of
$(\eps,k)$-RIP matrices with $m = O(\eps^{-2}k\log d\log^2 (k\log d))$
and which support matrix-vector multiplication in time $O(d\log d) +
m\cdot \log^{O(1)} d$.  When combined with \cite{KW11}, we obtain a JL
matrix with a number of rows $m = O(\eps^{-2}\log N\log d\log^2((\log
N) \log d)) =
O(\eps^{-2}\log N\log^3 d)$ and same embedding time. Thus for both
RIP and JL, our constructions support fast matrix-vector multiply
using the fewest rows known.

\bigskip

Our RIP and JL matrices maintain the $O(d\log d)$ running time of the
sampled
discrete Fourier matrix as long as $k < d/\polylog d$, and
never have multiplication time larger than $d\cdot\log^{O(1)} d$ even
for $k$ as large as $d$. Our results are given in \Figure{results}.

We remark that the restrictions $k \ge \polylog m$ in \Figure{results} can
be eliminated as long as $\eps$ is not too small, because in this case
 it is already known how to 
obtain optimal RIP matrices with fast multiplication for small $k$.
More precisely, the Fast Johnson-Lindenstrauss Transform of~\cite{AC09}, combined 
with~\cite{BDDW08}, give an $(\eps, k)$-RIP matrix with $m = O(\eps^{-2}k\log(d/k))$ rows that
supports matrix-vector multiplies in time $O(d\log d)$ as long a $k \le \eps^{2/3}d^{1/3}/\polylog d$.
Meanwhile, our restrictions in \Figure{results} require $k \ge \polylog m$.
Thus, the only case when neither our result nor the results of~\cite{AC09,BDDW08} applies occurs when
$\eps < (\polylog d)/\sqrt{d}$.  We note that when $\eps < 1/\sqrt{d}$, it is unknown how to obtain any $(\eps,k)$-RIP 
matrix with fewer than $d < 1/\eps^2$ rows, and this is already trivially obtained by the identity matrix.

\renewcommand*\arraystretch{1.5}
\begin{figure*}
\footnotesize
\begin{tabular}{|p{1in}|p{1.5in}|p{1.5in}|p{1in}|p{1in}|}
\hline
Ensemble & \# rows $m$ needed for RIP &   Matrix-vector \newline multiplication time  &
Restrictions & Reference \\
\hline
Partial Fourier & $O(\eps^{-2}k\log d \log^3 k)$ & $O(d\log d)$ &
&\cite{RV08, CGV13}\\
\hline
Partial Circulant & $O(\eps^{-2} k \log^2 d \log^2 k)$ &
$O(d\log m)$ & &\cite{KMR12} \\
\hline
Hash $\times$ \newline Partial Fourier & $O(\eps^{-2} k\log d\log^2 (k\log d))$ &
$O(d\log d) + m\polylog d$ & $k \geq \log^{2.5} m$ &
this work \\
\hline
Hash $\times$ \newline Partial Circulant & $O(\eps^{-2}k\log d\log^2 (k\log d))$ &
$O(d\log m) + m\polylog d$& $k \geq \log^2 m  $ & this work \\
\hline
\end{tabular}
\normalsize 
\caption{Table of results.}\FigureName{results}
\end{figure*}

\subsection{Notation and Preliminaries}
We set some notation. We use $[n]$ to denote the set $\{1,\ldots,n\}$.
 We use $\twonorm{\cdot}$
 denote the $\ell_2$ norm of a vector, and $\opnorm{\cdot}$,
 $\fronorm{\cdot}$ to denote the operator and Frobenius norms of a
 matrix, respectively.
 For a set $\mathcal{S}$ and a norm $\decnorm{\cdot}{X}$,
 $d_{\decnorm{\cdot}{X}}(\mathcal{S})$ denotes the
 diameter of $\mathcal{S}$ with respect to $\decnorm{\cdot}{X}$.  The
 set of $k$-sparse vectors $x \in \C^d$ with $\twonorm{x} \leq 1$ is
 denoted $T_k$.
 In addition to $O(\cdot)$ notation, for two functions $f, g$, we use
 the shorthand $f \lesssim g$ (resp. $\gtrsim$) to indicate
 that $f \leq Cg$ (resp. $\geq$) for some absolute constant $C$.  We
 use $f \eqsim g$ to mean $cf \leq g \leq Cf$ for some constants
 $c, C$.
 For clarity, we have made no attempt to optimize the values of the
 constants in our analyses.

Once we define the randomized construction of our RIP matrix $\Phi$, we will
control $|\| \Phi x\|_2^2 - \|x\|_2^2|$ uniformly over $T_k$, and thus will
need some tools for controlling the supremum of a stochastic process on a
compact set.  For a metric space $(T,d)$, the \em $\delta$-covering number \em
$\mathcal{N}(T, d, \delta)$ is the size of the smallest $\delta$-net of $T$ with
respect to the metric $d$. 
One way to control a stochastic process on $T$ is simply to union bound over a sufficiently fine net of $T$;
a more powerful way to control stochastic processes, due to Talagrand, is through the $\gamma_2$ functional~\cite{Talagrand05}.

\begin{definition}
For a metric space $(T,d)$, an \em admissible sequence \em of $T$ is a sequence of nets $A_1, A_2, \ldots$ of $T$ so that $|A_n| \leq 2^{2^n}$.  Then 
\[ \gamma_2(T, d) := \inf \sup_{t \in T} \sum_{n=1}^\infty 2^{n/2} d(A_n, t),\]
where the infimum is taken over all admissible sequences $\{A_n\}$.
\end{definition}

Intuitively, $\gamma_2(T,d)$ measures how ``clustered'' $T$ is with respect to $d$: if $T$ is very clustered, then the union bound over nets above can be improved by a chaining argument. 
A similar idea is used in Dudley's integral inequality~\cite[Theorem 11.1]{lt:1991}, and indeed they are related (see~\cite{Talagrand05}, Section 1.2) by
\begin{equation}\label{eq:dudleyize}
\gamma_2(T,d) \lesssim \int_0^{\text{diam}_d(T)} \sqrt{\log
  \mathcal{N}(T,d, u)}\,du.
\end{equation}
It is this latter form that will be useful to us.

\subsection{Organization}

In Section \ref{sec:tech} we define our construction and give an overview of our techniques.
We also state our most general theorem, Theorem \ref{thm:maingen}, which gives a recipe
for turning a ``good'' ensemble of RIP matrices into an ensemble of RIP matrices with fewer rows.
In Section \ref{sec:applications}, we apply Theorem \ref{thm:maingen} to obtain the results listed in Figure \ref{fig:results}. 
Finally, we prove Theorem \ref{thm:maingen} in Sections \ref{sec:outline} and \ref{sec:covering}.

\section{Technical Overview}
\label{sec:tech}
Our construction is actually a general method for turning any ``good''
RIP matrix with a suboptimal number of rows
into an RIP matrix with fewer rows.  Many previous constructions of RIP matrices involve beginning with an 
appropriately structured matrix (a DFT or Hadamard matrix, or a
circulant matrix, for example), and 
keeping only a subset of the rows.
In this work we propose a simple twist on this idea: each row of our
new matrix is a linear 
combination of a small number of rows from the original matrix, with
random sign flips as the coefficients.
Formally, we define our construction as follows.

Let $\mathcal{A}_M$ be a distribution on $M \times d$ matrices,
defined for all $M$, and fix parameters $m$ and $B$.  Define the
injective function $h:[m] \times [B] \to [mB]$ as $h(b, i) = B(b-1) +
i$ to partition $[mB]$ into $m$ buckets of size $B$, so $h(b, i)$
denotes the $i^{th}$ element in bucket $b$.  We draw a matrix $A$ from
$\mathcal{A}_{mB}$, and then construct our $m \times d$ matrix
$\Phi(A)$ by using $h$ to hash the rows of $A$ into $m$ buckets of
size $B$.
\begin{definition}[Our construction]\label{def:construction}
  Let $\mathcal{A}_M$ be as above, and fix parameters $m$ and $B$.
  Define a new distribution on $m \times d$ matrices by constructing a
  matrix $\Phi \in \C^{m \times d}$ as follows.
\begin{enumerate}
\item Draw $A \sim \mathcal{A}_{mB}$, and let $a_i$ denote the rows of $A$.
\item For each $(b,i) \in [m] \times [B]$, 
  choose a sign $\sigma_{b,i} \in \{\pm 1\}$ independently, uniformly at random.
\item For $b=1,\ldots,m$ let 
\[\varphi_b = \sum_{i \in [B]} \sigma_{b,i} a_{h(b,i)},\] 
and let $\Phi = \Phi(A,\sigma)$ be the matrix with rows $\varphi_b$.
\end{enumerate}
We use $A_b$ to denote the $B \times d$ matrix with rows $a_{h(b,i)}$
for $i \in [B]$.
\end{definition}
Equivalently, $\Phi$ may be obtained by writing $\Phi = HA$, where $A \sim \mathcal{A}_{mB}$,
and $H$ is the $m \times mB$ random matrix
with columns indexed by $(b,i) \in [m] \times [B]$, so that 
\[ H_{j,(b,i)} = \begin{cases} \sigma_{b,i} & b=j \\ 0 &b\neq j \end{cases}.\]
Note that there are two sources of randomness in the construction of $\Phi$: there is the choice of $A \sim \mathcal{A}_{mB}$,
and also the choice of the sign flips which determine the matrix $H$.  Our RIP matrix will be the appropriately normalized matrix $\Phi / \sqrt{mB}$. 

We consider two example distributions for $\mathcal{A}_M$.  
First, we consider a bounded orthogonal ensemble. 
\begin{definition}[Bounded orthogonal ensembles]
Let $U \in \mathbb{C}^{d \times d}$ be any unitary matrix with $|U_{ij}|
\leq 1$ for all entries $U_{ij}$ of $U$.
Let $u_i$ denote the $i^{th}$ row of $U$.
A matrix $A \in \mathbb{C}^{M \times d}$ is drawn from the \em bounded orthogonal ensemble \em associated with $U$
as follows. 
Select, independently and uniformly at random, a multi-set $\Omega = \{t_1, \ldots, t_M\}$ with $t_i \in [d]$.
Then let $A \in \mathbb{C}^{M \times d}$ be the matrix with rows $u_{t_1}, \ldots, u_{t_M}$.
\end{definition}
Popular choices (and our choices) for $U$ include the $d$-dimensional
discrete Fourier transform (resulting in the \em Fourier ensemble\em),
or the $d \times d$ Hadamard matrix, both of which support $O(d\log
d)$ time matrix-vector multiplication.

The second family we consider is the partial circulant ensemble.
\begin{definition}[Partial Circulant Ensemble]
For $z \in \mathbb{C}^d$, the \em circulant matrix \em $H_z \in \mathbb{C}^{d \times d}$
is given by $H_z x = z * x$, where $*$ denotes convolution.
Fix $\Omega \subset [d]$ of size $M$ arbitrarily.
A matrix $A$ is drawn from the \em partial circulant ensemble \em as follows. 
Choose $\eps \in \{\pm 1\}^d$ uniformly at random, and let $A$ be the 
rows of $H_\eps$ indexed by $\Omega$. 
\end{definition}

As long as the original matrix ensemble $\mathcal{A}$ supports fast matrix-vector multiplication, so does
the resulting matrix
$\Phi$. 
Indeed, writing $\Phi x = HAx$ as above, we observe that there are $mB$ nonzero entries in $H$, so computing the product $HAx$ takes time $O(mB)$, plus the time it takes to compute $Ax$.  When $A$ is drawn from the partial Fourier ensemble, $Ax$ may be computed in time $O(d\log d)$ via the fast Fourier transform.  We will choose $B = \polylog(d)$, and so $\Phi x$ may be computed in time $O(d\log d + m\polylog d)$.
When $A$ is the partial circulant ensemble, $Ax$ may be computed in time $d\log(mB)$ by breaking it up into $d/(mB)$ blocks, each of which is a $mB \times mB$ Toeplitz matrix supporting matrix-vector multiplication in time $O(mB\log(mB))$.  
Thus, in this case $\Phi x$ may be computed in time $O(d\log(mB) + mB) = O(d\log m) + m\polylog d$.

Having established the ``multiplication time" column of \Figure{results}, we 
turn to the more difficult task of establishing the bounds on $m$, the number of rows.
We note that $\Phi/\sqrt{mB}$ has the $(\eps, k)$-RIP if and only if
\[ \sup_{x \in T_k} \abs{ \frac{1}{mB}\twonorm{\Phi x}^2 - \twonorm{x}^2 } \leq \eps,\]
and so our goal will be to establish bounds on $\sup_{x \in T_k} \abs{ \twonorm{\Phi x}^2/(mB) - \twonorm{x}^2}$.
We will show that if $\mathcal{A}$ satisfies certain properties, then 
in expectation this quantity is small.
Specifically we require the following two conditions.
First, we require a random matrix from $\mathcal{A}$ to have the
RIP with a reasonable, though perhaps suboptimal, number of rows:
\begin{equation}\tag{$\star$}\label{eq:niceA}
 \E_{A \sim \mathcal{A}} \sup_{x \in T_k} \left| \frac{1}{M} \twonorm{A
     x}^2 - \twonorm{x}^2 \right| \lesssim \sqrt{\frac{L}{M}}
\end{equation}
for some quantity $L$, for suitably large $M > M_0$.

Second, the matrices $A_b$ whose rows are the rows of $A$ indexed by
$h(b,i)$ for $i \in [B]$ should be well-behaved.
Define~\eqref{eq:good} to be the event that
\begin{equation}\tag{$\dag$}\label{eq:good}
\max_{b \in [m]} \sup_{x \in T_s} \twonorm{ A_b x } \leq \ell(s)
\end{equation}
for some function $\ell(s)$ and all $s \leq 2k$.  We require
that~\eqref{eq:good} happen with constant probability:
\begin{equation}\tag{$\star\star$}\label{eq:concentration}
\Pr_{A \sim \mathcal{A}} \left[ \text{\eqref{eq:good} holds}  \right] \geq 7/8.
\end{equation}
for some sufficiently small function $\ell$.

As long as these two requirements on $\mathcal{A}$ are satisfied, 
and all matrices in the support of $\mathcal{A}$ have entries of
bounded magnitude, the
construction of Definition \ref{def:construction}
yields a RIP matrix, with appropriate parameters.  The
following is
our most general theorem.
\begin{theorem}\label{thm:maingen}
Fix $\eps \in (0,1)$, and fix integers $m$ and $B$. 
Let $\mathcal{A} = \mathcal{A}_{mB}$ be a distribution on $mB
\times d$ matrices so that $\inftynorm{a_i} \leq 1$ almost surely for
all rows $a_i$ of $A \sim \mathcal{A}$.
Suppose  that \eqref{eq:niceA} holds with
\[ L \leq mB \eps^2,\]
and $M = mB > M_0$.
Suppose further that
 \eqref{eq:concentration} holds, with 
\[ \ell(s) \leq Q_1 \sqrt{B} + Q_2 \sqrt{s}\]
and that \[ B \geq \max\{ Q_2^2 \log^2 m, Q_1^2 \log m
\log k \}, \hbox{ and }  k \geq Q_1^2 \log^2 m . \]
Finally, suppose that $m > m_0$, for
\[ m_0 = \frac{k\log d\log^2(Bk) }{\eps^2}.\]
Let $\Phi$ be drawn from the distribution of Definition
\ref{def:construction}.  Then
\[
\sup_{x \in T_k} \left| \frac{1}{mB} \twonorm{\Phi x}^2 -
  \twonorm{x}^2 \right| \lesssim \eps,
\]
that is, $\frac{1}{\sqrt{mB}}\Phi$ satisfies the $(O(\eps),k)$-RIP, with $3/4$ probability.
\end{theorem} 

In Section \ref{sec:applications}, we will show how to use Theorem \ref{thm:maingen}
to prove the results reported in \Figure{results}, but first we will outline the intuition of the proof of Theorem \ref{thm:maingen}.

By construction, the expectation of $\twonorm{\Phi x}^2$ over the sign
flips $\sigma$ is simply $\twonorm{Ax}^2$, and 
\eqref{eq:niceA} guarantees that this expectation is under
control, uniformly over $x \in T_k$.  The trick is 
that $A$ has $mB$ rows, rather than $m$, and this provides slack
to handle the fact that the guarantee \eqref{eq:niceA} 
is not optimal. 

The problem is then to argue that for all $x \in T_k$, $\twonorm{\Phi x}^2$ is close to its expectation.
The proof of Theorem \ref{thm:maingen} proceeds in two steps.  First, we condition on $A$ and control the deviation 
\begin{equation}\label{eq:tocontrol}
\E_\sigma \sup_{x \in T_k} \left| \twonorm{\Phi x}^2 - \E_\sigma \twonorm{\Phi x}^2 \right|.
\end{equation}
Second, we take the expectation with respect to $A \sim \mathcal{A}_{mB}$.

In Theorem \ref{thm:boundbygamma2} 
we carry out the first step and bound the deviation
\eqref{eq:tocontrol} by Talagrand's $\gamma_2$ functional 
$\gamma_2(T_k,\decnorm{\cdot}{X})$, where $\decnorm{x}{X} := \max_b
\twonorm{A_b x}$ 
is a norm which measures the contribution to $\twonorm{\Phi x}$ of the worst bucket
$b$ of the partition function $h$.
Our strategy is to write $\twonorm{\Phi x}^2$ as $\twonorm{X(x)
  \sigma}^2$, for an appropriate matrix $X(x)$ that depends on $A$.  
Finally we use a result of Krahmer, Mendelson, and Rauhut~\cite{KMR12}
to control the Rademacher chaos, obtaining an expression 
in terms of $\gamma_2(T_k, \decnorm{\cdot}{X})$. 

In the second step, we unfix $A$, and $\gamma_2(T_k,
\decnorm{\cdot}{X})$ becomes a random variable.  In Theorem
\ref{thm:covering}, we show that, as long as \eqref{eq:concentration}
holds, $\gamma_2(T_k,\decnorm{\cdot}{X})$ is small with high
probability over the choice of $A \sim \mathcal{A}_{mB}$.  By
\eqref{eq:dudleyize}, it is sufficient to bound the covering numbers
$\mathcal{N}(T_k, \decnorm{\cdot}{X}, u)$. This is similar
to~\cite{RV08}, which must bound the same $\mathcal{N}(T_k,
\decnorm{\cdot}{X}, u)$ but in a setting where $B = 1$.  Both papers
use Maurey's empirical method to relate the covering number to
$\E[\max_b \decnorm{A_bg}{2}]$ for a Gaussian process $g$.  But
while~\cite{RV08} loses a $\sqrt{\log m}$ factor in a union bound over
$b$, we only lose a constant factor as long as $B \geq \polylog d$.
This difference is what gives our $\log k$ improvement in $m$.
It is also the most technical piece of our proof, and
is presented in Section \ref{sec:covering}.  

Finally, we put all of the pieces together.
As long as $mB$ is large enough and the condition \eqref{eq:niceA} holds,
$\E_\sigma \twonorm{\Phi x}/\sqrt{mB}$ will be close to
$\twonorm{x}$ in
expectation over $A$.  At the same time
as long as the condition \eqref{eq:concentration} holds, 
the deviation \eqref{eq:tocontrol} is small in expectation over $A \sim \mathcal{A}_{mB}$.
Choosing $B$ appropriately controls the restricted isometry constant
of $\Phi$, at the cost of slightly increasing the embedding time.


\section{Main Results}\label{sec:applications}
Before we prove Theorem \ref{thm:maingen}, let us show how we may use it to conclude the results in \Figure{results}.
To do this, we must compute $L$ and $\ell(s)$ from the conditions \eqref{eq:niceA} and \eqref{eq:concentration}, 
when
$\mathcal{A}$ is the Fourier ensemble (or any bounded orthogonal ensemble), and when $\mathcal{A}$ is 
the partial circulant ensemble.

\subsection{Bounded orthogonal ensembles}
\label{sec:fourier}
Suppose $\mathcal{A}$ is a bounded orthogonal ensemble.
The RIP analysis of~\cite{RV08,CGV13} shows 
\[ \E_{A \sim \mathcal{A}} \sup_{x \in T_k} \left| \frac{1}{M}
  \twonorm{A x}^2 - \twonorm{x}^2 \right| \lesssim \sqrt{\frac{ k
      \log^3 k \log d}{M} },\]
provided that $M \gtrsim k\log^3 k \log d$,
so we may take $L \lesssim k \log^3 k \log d$.
Further, the analysis of~\cite{RV08} (see Lemma \ref{lemma:largeRV})
implies that
\[\Pr_{A \sim \mathcal{A}} \left[  \exists s\in[2k] :  \max_{b \in [m]} \sup_{x \in T_s}
  \twonorm{ A_b x } \geq \ell(s) \right] 
  \leq 2km \max_{s\in[2k]} \Pr_{A \sim \mathcal{A}} \left[  \sup_{x \in T_s}
  \twonorm{ A_1 x } \geq \ell(s) \right] \leq 1/8\]
when 
\[ \ell(s) \eqsim  \log^{1/4}(m)\sqrt{B} + \log^{1/4}(m) \sqrt{s \log^2 (k) \log(d) \log(B)}.\]
Thus, we may take $Q_1 \lesssim \log^{1/4} m$ and $Q_2 \lesssim \log^{1/4}(m) \log(k) \sqrt{\log(d)\log(B)} \lesssim \log^{2.5}(d)$
With these parameter settings,
Theorem \ref{thm:maingen} implies the
following theorem.
\begin{theorem}\label{thm:fourier}
  Let $\eps \in (0,1)$.  Let $\mathcal{A}$ be a bounded orthogonal
  ensemble (for example, the Fourier ensemble), and suppose that
  $\Phi$ is as in Definition \ref{def:construction}.  Further suppose
  $B \geq \log^{6.5} d$ and $k \geq \log^{2.5} m$.  Then for some
  value
  \[ m = O\left( \frac{ k \log d \log^2 (k\log d)}{\eps^2} \right), \]
  we have that
  \[ \sup_{x\in T_k}\left| \twonorm{\Phi x}^2 - \twonorm{x}^2 \right| \leq \eps \]
  with $3/4$ probability.
\end{theorem}

\subsection{Circulant Matrices}
\label{sec:circulant}
Suppose that $\mathcal{A}$ is the partial circulant ensemble.
By the analysis in~\cite{KMR12}, 
\[ \E_{A \sim \mathcal{A}} \sup_{x \in T_k} \left| \frac{1}{M}
  \twonorm{A x}^2 - \twonorm{x}^2 \right| \lesssim \sqrt{\frac{ k
    \log^2 k \log^2 d }{M}},\]
for $M \gtrsim k\log^2 k\log^2 d$.
Concentration also follows from the analysis in~\cite{KMR12}, as a
corollary of Theorem \ref{thm:kmr} (see \cite[Theorem 4.1]{KMR12}). 
\begin{lemma} (Implicit in~\cite{KMR12})\LemmaName{circ-conc}
\[\Pr_{A \sim \mathcal{A}} \left[  \exists s\in[2k] :  \max_{b \in [m]} \sup_{x \in T_s}
  \twonorm{ A_b x } \geq \ell(s) \right] \leq \frac{1}{8}\]
when 
\[ \ell(s) \eqsim \sqrt{B} + \sqrt{s} \log k \log d.\]
\end{lemma}
Thus, we may take $Q_1 \lesssim 1$ and $Q_2 \lesssim \log k\log d$. Then Theorem \ref{thm:maingen} implies the following theorem.
\begin{theorem}\label{thm:circulant}
  Let $\eps \in (0,1)$.  Let $\mathcal{A}$ be the partial circulant
  ensemble, and suppose $\Phi$ is constructed as in Definition
  \ref{def:construction}.  Further suppose $B \geq \log^2 m\log^2
  k\log^2 d$ and $k \geq \log^2 m$.  Then, for some value
  \[m = O\left( \frac{k \log d \log^2 (k\log d)}{\eps^2} \right), \] we
  have that, as long as $m < d/B$,
  \[ \sup_{x\in T_k} \left| \twonorm{\Phi x}^2 - \twonorm{x}^2
  \right| \leq \eps \]
  with $3/4$ probability.
\end{theorem}

We remark that the condition $m \leq d/B$ does not actually effect the
results reported in \Figure{results}.  Indeed, if $mB > d$, we may
artificially increase $d$ to $d' = mB$ by embedding $T_k$ in $\C^{d'}$
by zero-padding.  Applying Theorem \ref{thm:circulant} with $d = d'$
implies an RIP matrix with $O(\eps^{-2}k\log d' \log^2 (k \log d))$ rows
and embedding time $O(d'\log d') + m\polylog d'$.  Because $B =
\polylog d$, we have $d' = d\polylog(d)$, and there is no asymptotic
loss in $m$ by extending $d$ to $d'$.  Further, in this parameter
regime, $d' \log d' = mB\log d' = m \polylog d$.

\section{Proof of Theorem \ref{thm:maingen}}
\label{sec:outline}
We will use the following theorem from~\cite{KMR12}.
\begin{theorem}{\cite[Theorem 1.4]{KMR12}}\label{thm:kmr}
Let $\mathcal{S} \subset \mathbb{C}^{m \times M}$ be a symmetric set
of matrices, $\mathcal{S} = -\mathcal{S}$.  Let $\sigma \in \{\pm 1\}^M$ uniformly at random.  Then
\[ \E \sup_{X \in \mathcal{S}} \left| \twonorm{X \sigma}^2 - \E
  \twonorm{X\sigma}^2 \right| \lesssim \left( d_F(\mathcal{S})
  \gamma_2(\mathcal{S}, \opnorm{\cdot} ) + \gamma_2^2(\mathcal{S},
  \opnorm{\cdot}) \right) =: E'.\]
Furthermore, for all $t > 0$,
\[ \Pr\left[ \sup_{X \in \mathcal{S}} \left| \twonorm{X\sigma}^2 - \E
    \twonorm{X\sigma}^2 \right| > C_1 E' + t \right] \leq
2\exp\left(-C_2\min\left\{ \frac{t^2}{V^2}, \frac{t}{U} \right\}
\right),\]
where $C_1$ and $C_2$ are constants, 
\[ V = d_{2\to 2}(\mathcal{S})( \gamma_2( \mathcal{S}, \opnorm{\cdot}
) + d_F(\mathcal{S})), \]
and
\[ U = d_{2 \to 2}^2(\mathcal{S}).\]
\end{theorem}

The first step in proving Theorem \ref{thm:maingen} is to bound the
restricted isometry constant of $\Phi$ in terms of the $\gamma_2$
functional, removing the dependence on $\sigma$.
\begin{theorem}\label{thm:boundbygamma2}
Suppose $\mathcal{A} = \mathcal{A}_M$ is a distribution on $M
\times d$ matrices so that \eqref{eq:niceA}  
holds, and let $\Phi$ be as in Definition \ref{def:construction}. Then
\begin{equation}
\E \sup_{x \in T_k} \left| \frac{1}{mB} \twonorm{\Phi x}^2 -
  \twonorm{x}^2\right| \lesssim \frac{1}{mB} \left( \E_A \sup_{x \in T_k}
  \twonorm{Ax} \gamma_2( T_k, \decnorm{\cdot}{X} ) + \E_A \gamma_2^2(T_k,
  \decnorm{\cdot}{X}) \right) + \sqrt{ \frac{ L }{mB}}.
\end{equation}
where 
\[\decnorm{x}{X} := \max_{b \in [m]} \twonorm{A_b x}.\]
\end{theorem}
\begin{proof}
Let $H(b) = \setst{ h(b,i) }{i \in [B]}$ be the multiset of indices of
the rows of $A$ in bucket $b$,  and as above let $A_b$ denote the
$B\times d$ matrix whose rows are indexed by $H(b)$.  Let $\sigma_b
= \sum_{i=1}^B \sigma_{b,i} e_i$ denote the vector of sign flips
associated with bucket $b$. 
Notice that, by construction, conditioning on $A \sim \mathcal{A}$, we have
\begin{equation}\label{eq:goodExptmp}
 \E_\sigma \twonorm{\Phi x}^2 = \twonorm{Ax}^2,
\end{equation}
and so 
\begin{align}
&\mathbb{E} \sup_{x \in T_k} \left| \frac{1}{mB} \twonorm{\Phi x}^2 -
  \|x\|^2\right| \notag \\
&\qquad \leq
\mathbb{E}_{A} \left[ \frac{1}{mB} \mathbb{E}_\sigma  \sup_{x \in T_k}
    \left| \twonorm{\Phi x}^2 - \mathbb{E}_\sigma \twonorm{\Phi x}^2 \right| 
         + \sup_{x \in T_k}
  \left| \frac{1}{mB} \mathbb{E}_\sigma \twonorm{\Phi x}^2  -
    \twonorm{x}^2 \right| \right]
\notag \\
&\qquad = \frac{1}{mB} \mathbb{E}_{A} \mathbb{E}_\sigma \sup_{x \in T_k}
    \left| \twonorm{\Phi x}^2 - \twonorm{A x}^2 \right|  +
    \mathbb{E}_A \sup_{x \in T_k} \left|
  \frac{1}{mB} \twonorm{A x}^2 - \twonorm{x}^2 \right| \notag \\
&\qquad \lesssim \frac{1}{mB} \mathbb{E}_{A} \mathbb{E}_\sigma \sup_{x
  \in T_k} \left| \twonorm{\Phi x}^2 - \twonorm{A x}^2 \right|  +
    \sqrt{\frac{L}{mB}},
\label{eq:theplan}
\end{align}
where we have used \eqref{eq:niceA} in the last line and
\eqref{eq:goodExptmp} in the penultimate line.

Condition on the choice of $A$ until further
notice, and consider the first term. We may write
\[ E := \E_\sigma \sup_{x \in T_k} \left|\twonorm{ \Phi x }^2 - \E_\sigma
  \twonorm{ \Phi x}^2\right| =  \E_\sigma \sup_{x \in T_k} \left| \sum_b
  \abs{\ip{\sigma_b}{A_b x}}^2 - \E_\sigma \sum_b \abs{\ip{\sigma_b}{A_b x}}^2
\right|.\]

Now, we apply Theorem \ref{thm:kmr} to $\mathcal{S} =  \{ X(x) \in \mathbb{C}^{m
  \times mB} \mid x \in T_k\}$, where $X(x)$ is defined as follows:
\[ X(x) =  \begin{bmatrix} - (A_1 x)^* - & 0 & 0 & \cdots & 0\\
0 & - (A_2 x)^* - & 0 & \cdots & 0\\
0 &0& - (A_3 x)^* - & \cdots & 0\\
\vdots &\vdots & \vdots &  &\vdots\\
0&0&0&\cdots& -(A_m x)^* -
\end{bmatrix}.\]
Let $\sigma$ be the vector in $\{-1,1\}^M$ defined as
$(\sigma_1^*,\ldots,\sigma_m^*)^*$. By construction, $\twonorm{X(x)
  \sigma}^2 = \sum_b \abs{\ip{\sigma_b}{A_b
  x}}^2$, and so by Theorem \ref{thm:kmr}, it suffices to control
$d_F(\mathcal{S})$ and $\gamma_2(\mathcal{S}, \opnorm{\cdot})$. The
Frobenius norm of $X(x)$ is 
\[ \fronorm{X(x)}^2 = \sum_{b \in [m]} \twonorm{A_b x}^2 =
\twonorm{Ax}^2.\]
For the $\gamma_2$ term, notice that for any $x,y \in T_k$, 
\[ \opnorm{ X(x) - X(y) } = \max_{b \in [m]} \|A_b(x-y)\|_2 =
\decnorm{x-y}{X}.\]
Thus, $\gamma_2( \mathcal{S}, \opnorm{\cdot} ) = \gamma_2( T_k,
\decnorm{\cdot}{X} )$. Then Theorem \ref{thm:kmr} implies that
\[ E \lesssim \max_{x \in T_k} \twonorm{Ax} \gamma_2( T_k,
  \decnorm{\cdot}{X} ) + \gamma_2^2(T_k, \decnorm{\cdot}{X}).\]
Plugging this into \eqref{eq:theplan}, we conclude
\begin{equation}
\E \sup_{x \in T_k} \left| \frac{1}{mB} \twonorm{\Phi x}^2 -
  \twonorm{x}^2\right| \lesssim \frac{1}{mB} \left( \E_A \sup_{x \in T_k}
  \twonorm{Ax} \gamma_2( T_k, \decnorm{\cdot}{X} ) + \E_A \gamma_2^2(T_k,
  \decnorm{\cdot}{X}) \right) + \sqrt{ \frac{ L }{mB}}.
\end{equation}
\end{proof}

Theorem \ref{thm:boundbygamma2} leaves us with the task of controlling $\gamma_2(T_k, \decnorm{\cdot}{X})$,
which we do in the following theorem.

\begin{theorem}\label{thm:covering}
Suppose that $A$ is a matrix such that \eqref{eq:good} holds, with 
\[ \ell(s) \leq Q_1 \sqrt{B} + Q_2 \sqrt{s}.\]
Suppose further that $\inftynorm{a_i} \leq 1$ for all
$i$, and suppose that \[ B \geq \max\{ Q_2^2 \log^2 m, Q_1^2 \log m
\log k \}, \hbox{ and }  k \geq Q_1^2 \log^2 m . \]
Then
\[ \gamma_2(T_k, \decnorm{\cdot}{X}) \lesssim \sqrt{kB\log d} \cdot
\log (Bk).\]
\end{theorem}

\begin{proof}
By \eqref{eq:dudleyize},
\begin{equation}\label{eqn:dudley}
	 \gamma_2(T_k, \Xnorm{\cdot} ) \lesssim \int_{u = 0}^Q \sqrt{\log \mathcal{N}( T_k,  \Xnorm{\cdot}, u ) }\, du,
\end{equation}
where $Q = \sup_{x \in T_k} \Xnorm{x}$.  
Notice that we can bound
\[
Q^2 = \sup_{x \in T_k} \max_b \twonorm{A_bx}^2= \sup_{x \in T_k} \max_b
\sum_{i \in [B]} \abs{\ip{a_{h(b, i)}}{x}}^2 \leq B \sup_{x \in T_k} \|x\|_1^2 \leq Bk
\]
using the fact that each entry of $a_{h(b,i)}$ has magnitude at most $1$.  We
follow
the approach of~\cite{RV08} and estimate the covering number using two
nets, one for small $u$ and one for large $u$.  

For small $u$, we use a standard $\ell_2$ net of $B_2$:  
we have
\[ \Xnorm{x} \leq Q \twonorm{x}\] so $\mathcal{N}(T_k, \Xnorm{\cdot},
u) \leq \mathcal{N}(T_k, \twonorm{\cdot}, u/Q)$.  Observing that $T_k$ is the union
of $\binom{d}{k} = \left( \frac{d}{k} \right)^{O(k)}$ copies of
$B_2^k$ (the unit $\ell_2$-ball of dimension $k$), we may cover $T_k$ by covering 
cover each copy of $B_2^k$
with a net of width $u/Q$.  By a standard volume estimate
\cite[Eqn.\ (5.7)]{Pisier89},
the size of each such net is $(1 + 2Q/u)^k$, and so
\[
\sqrt{\log \mathcal{N}(T_k, \Xnorm{\cdot}, u)} \lesssim \sqrt{k
  \log (d/k) + k \log (1 + 2Q/u)} \lesssim \sqrt{ k
  \log(dQ/u)}.
\]
For
large $u$ the situation is not as simple.
We show in Lemma \ref{lem:covering} that, 
as long as \eqref{eq:good} holds,
\[
\sqrt{ \log \mathcal{N}( T_k, \Xnorm{\cdot}, u)} 
\lesssim \frac{ \sqrt{kB\log d}}{u}.
\] 

We plug these bounds into \eqref{eqn:dudley} and integrate, using the
first net for $u \in (0,1)$ and the second for $u > 1$.
We find
\begin{align*}
\int_{u = 0}^Q \sqrt{\log \mathcal{N}(T_k,  \max_b \|F_b \cdot \|, u) }\, du 
&\lesssim \int_{u = 0}^1  \sqrt{k \log(dQ/u)}\,
du + \int_{u = 1}^Q \frac{\sqrt{ kB \log d }}{u} \,du \\
&\lesssim \sqrt{k\log(dQ)} + \sqrt{kB \log d } \log Q\\
&\lesssim \sqrt{kB\log d}\log Q\\
&\leq \sqrt{kB\log d}\log(Bk)
\end{align*}
as claimed.
\end{proof}

It remains to put Theorem \ref{thm:boundbygamma2} and Theorem
\ref{thm:covering} together to prove Theorem \ref{thm:maingen}.

\begin{proof} (Proof of Theorem \ref{thm:maingen}.)  We need to show that
  \[
  \Delta := \sup_{x \in T_k} \left| \frac{1}{mB} \twonorm{\Phi
    x}^2 - \twonorm{x}^2 \right| \lesssim \eps
  \]
  with $3/4$ probability.  
  We have by~\eqref{eq:concentration}
  that~\eqref{eq:good} holds with $7/8$ probability over $\mathcal{A}$, and we will show that
   $\Delta \lesssim \eps$ with $7/8$
  probability when $A$ is drawn from the distribution $\mathcal{A}' =
  (\mathcal{A} \mid \text{\eqref{eq:good} holds})$.  
  Together, this will imply the conclusion of Theorem \ref{thm:maingen}.

Note that as long as \eqref{eq:niceA} holds for $\mathcal{A}$, \eqref{eq:niceA} holds for $\mathcal{A}'$ as well.  Indeed,  
  \[
  \E_{A \sim \mathcal{A}'} \sup_{x \in T_k} \left| \frac{1}{mB}
    \twonorm{A x}^2 - \twonorm{x}^2 \right| \leq \left(\frac{8}{7}\right)\E_{A \sim
    \mathcal{A}} \sup_{x \in T_k} \left| \frac{1}{mB} \twonorm{A x}^2 -
    \twonorm{x}^2 \right| \lesssim \eps,
  \]
  so~\eqref{eq:niceA} holds for $\mathcal{A}'$. 
  For the rest of the proof, we consider $A \sim \mathcal{A}'$, so we have
\[
	\frac{1}{\sqrt{mB}} \E_A \sup_{x \in T_k} \twonorm{Ax}  \leq \sqrt{1 + O(\eps)} \lesssim 1.
\]
  Under the parameters of Theorem~\ref{thm:maingen} and
  because~\eqref{eq:good} holds for all $A \sim \mathcal{A}'$,
  Theorem~\eqref{thm:covering} implies
  \[ \gamma_2(T_k, \decnorm{\cdot}{X}) \leq \sqrt{kB\log d} \cdot
  \log (Bk).\]
Then
\begin{align*}
 \frac{1}{mB} \E_A\left[ \sup_{x \in T_k} \twonorm{Ax} \cdot \gamma_2(T_k, \decnorm{\cdot}{X} ) \right]
&\lesssim \frac{\sqrt{k\log(d)} \cdot \log(Bk)}{\sqrt{m}} \leq \eps.
\end{align*}
Similarly,
\begin{align*}
 \frac{1}{mB} \E_A \gamma_2^2(T_k, \decnorm{\cdot}{X}) &\lesssim \frac{ k\log(d) \log^2(Bk) }{m} \leq \eps^2.
\end{align*}
By Theorem \ref{thm:boundbygamma2}, and using the above bounds,
\begin{align*}
	\E_A[\Delta] &\lesssim \frac{1}{mB} \left( \E_A \sup_{x \in T_k}
          \twonorm{Ax} \gamma_2( T_k, \decnorm{\cdot}{X} ) + 
		\E_A \gamma_2^2(T_k, \decnorm{\cdot}{X}) \right) +
              \sqrt{ \frac{ L }{mB}}\\
	&\lesssim \eps + \eps^2 + \eps \\
	&\lesssim \eps.
\end{align*}
Therefore by Markov's inequality,
we have $\Delta \lesssim \eps$ with arbitrarily high constant probability over $A \sim \mathcal{A}'$.
In particular, we may adjust the constants so that $\Delta \lesssim \eps$ with probability at least $7/8$ over
$A \sim \mathcal{A}'$, which was our goal.
\end{proof}

\newcommand{\enorm}[2]{\left \lVert#2\right \rVert_{#1}}

\section{Covering number bound}
\label{sec:covering}
In this section, we prove the covering number lemma needed for the proof of Theorem \ref{thm:covering}.
Recall the definition $\enorm{X}{x} = \max_{b \in [m]} \enorm{2}{A_b
  x}$, and that $T_k$ is the set of $k$-sparse vectors in $\C^d$ with
$\ell_2$ norm at most $1$.

\begin{lemma} \label{lem:covering} 
Suppose that 
the conditions of Theorem \ref{thm:covering} hold.  Then 
  \[
  \N(T_k/\sqrt{k}, \enorm{X}{\cdot}, u) \leq (2d+1)^{O(B/u^2)}.
  \]
\end{lemma}

We will prove this under the assumption that $x \in T_k$ is real,
using only that~\eqref{eq:good} holds for $s \leq k$ and that $A$ has
bounded entries.  Then by Proposition~\ref{prop:go-real} in the
Appendix, we have $\N(T_k/\sqrt{k}, \enorm{X}{\cdot}, u)$ over the
complex numbers is less than $\N(T_{2k}/\sqrt{2k}, \enorm{\tilde{X}}{\cdot},
u)$ over the reals, where $\enorm{\tilde{X}}{\cdot}$ denotes a version
of the $\enorm{X}{\cdot}$ for a matrix $\tilde{A}$ of bounded entries that
satisfies~\eqref{eq:good} for $s \leq 2k$.  
Adjusting the constants by a factor of $2$ gives the final result. 

As in \cite{RV08}, we use Maurey's empirical method (see
\cite{Carl85}).  Consider $x \in T_k/\sqrt{k}$, and choose a parameter $s$.  
For $i \in [s]$, define a random variable $Z_i$, so that $Z_i = e_j \text{sign}(x_j)$ with
probability $\abs{x_j}$ for all $j \in [d]$, and $0$ with probability
$1-\|x\|_1$. Notice that by the assumption that $x$ is real, $\text{sign}(x_j)$ is well defined.
Further, because $T_k/\sqrt{k} \subset B_1$, this is a valid
probability distribution.  We want to show for every $x$ that
\begin{align}
  \E\enorm{X}{x - \frac{1}{s} \sum Z_i} \lesssim
  \sqrt{\frac{B}{s}}. \EquationName{ineq-goal}
\end{align}
This would imply that the right hand side is at most $u$ for $s
\lesssim B/u^2$.  If this holds, then the set of all possible
$\frac{1}{s} \sum Z_i$ forms a $u$-covering.  As there are only $2d +
1$ choices for each $Z_i$, there are only $(2d + 1)^s$ different
vectors of the form $\frac{1}{s} \sum_{i=1}^s Z_i$.  These form a
$u$-covering, so \Equation{ineq-goal} will imply
\[\N(T_k, \enorm{X}{\cdot},u) \leq (2d + 1)^{O(B/u^2)}.\]

We now show \Equation{ineq-goal}. Draw a Gaussian vector $g \sim N(0, I_s)$, and define
\[
\G(x) = \E_{g,Z} \enorm{X}{\sum Z_ig_i}
\]
By a standard symmetrization argument followed by a comparison
principle (Lemma 6.3 and Eq.~(4.8) in~\cite{lt:1991} respectively, or
the proof of Lemma 3.9 in \cite{RV08}), 
\[
\E\enorm{X}{x - \frac{1}{s} \sum Z_i} \lesssim
\frac{1}{s}\G(x),
\]
so it suffices to bound $\G(x)$ by $O(\sqrt{Bs})$.

Let $L = \{i : \abs{x_i} > \frac{\log m}{k}\}$ be the set of
coordinates of $x$ with ``large'' value in magnitude.  Then
\[
\G(x) \leq \G(x_L) + \G(x_{\overline{L}})
\]
by partitioning the $Z_i$ into those from $L$ and those from
$\overline{L}$ and applying the triangle inequality.  
Notice that $x_L$ is ``spiky" and $x_{\overline{L}}$ is ``flat:" more precisely, we have 
\begin{equation}\label{eq:spikyflat}
\enorm{1}{x_L} \leq
\frac{1}{\log m} \qquad \text{and} \qquad \enorm{\infty}{x_{\overline{L}}} \leq
\frac{\log m}{k},
\end{equation}
using Cauchy-Schwarz to bound the $\ell_1$ norm.
To bound $\G(x_L)$ and $\G(x_{\overline{L}})$ we use the
following lemma.

\begin{lemma}\label{l:infl1bound}
  Suppose that \eqref{eq:good} holds.  Then the
	following inequalities hold for all $x$:
  \begin{align}
    \G(x) &\lesssim \sqrt{Bs \enorm{1}{x} \log
        m} \label{e:l1bound}\\
    \G(x) &\lesssim \sqrt{Bs} +  \sqrt{\log m} \left( Q_1\sqrt{B} +
      Q_2\sqrt{\min(k,s)} \right) \sqrt{ s\|x\|_\infty + \log k }  \label{e:linfbound}
  \end{align}
\end{lemma}
\begin{proof}
  \newcommand{\ZZ}{\mathbf{Z}} Let $\ZZ \in \{-1, 0, 1\}^{d \times s}$
  have columns $Z_i$, and $Z = \sum_i Z_i$.  Then
  \[
  \G(x) = \E \max_{b \in [m]} \enorm{2}{A_b \ZZ g}
  \]
  Consider $\enorm{2}{A_b \ZZ g}$ for a single $b \in [m]$.  This is a
  $C$-Lipschitz function of a Gaussian for $C = \enorm{2\to 2}{A_b \ZZ}$.
  Therefore~\cite[Eq. (1.4)]{lt:1991},
  \[
  \Pr_g[\enorm{2}{A_b \ZZ g} > \E_g\enorm{2}{A_b \ZZ g} + t\enorm{2\to 2}{A_b
    \ZZ}] < e^{-\Omega(t^2)}.
  \]
  Hence by a standard computation for subgaussian random variables~\cite[Eq. (3.13)]{lt:1991}), 
  \[
  \G(x) \lesssim \E_Z \max_{b \in [m]} \E_g\enorm{2}{A_b \ZZ g} + \sqrt{\log
    m}\enorm{2\to 2}{A_b \ZZ}.
  \]
  Now,
  \begin{align}\label{e:frob1}
    \E_g\enorm{2}{A_b \ZZ g} \leq \sqrt{\E_g\enorm{2}{A_b \ZZ g}^2} =
    \enorm{F}{A_b \ZZ} = \sqrt{B \enorm{1}{Z}}
  \end{align}
  and
  \begin{align}\label{e:frob2}
    \E_Z \sqrt{B \enorm{1}{Z}} \leq \sqrt{B \E_Z
      \enorm{1}{Z}} = \sqrt{B s \enorm{1}{x}} \leq \sqrt{Bs}.
  \end{align}
  Thus
  \begin{align}
    \G(x) \leq \sqrt{Bs\enorm{1}{x}} + O\left(\E_Z \max_{b \in [m]} \sqrt{\log
      m}\enorm{2\to 2}{A_b \ZZ}\right).
    \label{e:Gop}
  \end{align}
  Thus it suffices to bound $\enorm{2\to 2}{A_b \ZZ}$ in terms of
  $\enorm{1}{x}$ and $\enorm{\infty}{x}$.  First, we have
  \[
  \enorm{2\to 2}{A_b \ZZ} \leq \enorm{F}{A_b \ZZ}
  \]
  and so by Equations~\eqref{e:frob1} and~\eqref{e:frob2} we have
  \[
  \G(x) \leq \sqrt{Bs\enorm{1}{x}\log m},
  \]
  as desired for Equation~\eqref{e:l1bound}.

Second, we turn to Equation~\eqref{e:linfbound}.  For a matrix $A \in m \times d$ and a set $S \subset [d]$, let $\left.A\right|_S$ denote the $m \times d$ matrix with all the columns not indexed by $S$ set to zero.
  Then, we have
  \begin{align}\label{e:tolinf}
    \enorm{2\to 2}{A_b \ZZ} \leq \enorm{2\to 2}{\left.A_{b} \right|_{\supp(Z)}}
    \enorm{2\to 2}{\ZZ} \leq \max_{\abs{S} \leq \min(k, s)} \enorm{2\to 2}{\left.A_{b}\right|_{ S}} \enorm{\infty}{Z}^{1/2}.
  \end{align}
In the final step, we used the fact that for any matrix $A$, $\|A\|_{2
  \to 2} \leq \sqrt{\|A\|_{1 \to 1} \|A\|_{\infty \to \infty}}$ (see
\Lemma{operator-inequalities} in the Appendix).
By the assumption \eqref{eq:good} and the choice of $\ell$, 
\[ \max_{b \in [m]} \sup_{x \in T_{\min(k,s)}}
  \twonorm{ A_b x } \leq Q_1 \sqrt{B} + Q_2 \sqrt{\min(k,s)},\] 
so
\[ \max_{b \in [m]} \enorm{2 \to 2}{A_b\ZZ} \leq \enorm{\infty}{Z}^{1/2} \left( Q_1 \sqrt{B} + Q_2 \sqrt{\min(k,s)} \right).\]
  Finally, we bound $\E_Z \enorm{\infty}{Z}$.  By a Chernoff bound,
  for any $j \in \supp(x)$, we have
 \[ \prob{ \abs{ \left(\sum Z_i\right)_j } > s\abs{x_j}  + t} \leq
 e^{-\Omega(t)}. \]
  Integrating, we have
  \[ \E \enorm{\infty}{Z} \lesssim s\enorm{\infty}{x} + \log k.\]  
Thus
\[ \E_Z \max_{b \in [m]} \enorm{2 \to 2}{A_b\ZZ} \lesssim \left( s\enorm{\infty}{x} + \log k \right)^{1/2} \left( Q_1 \sqrt{B} + Q_2 \sqrt{\min(k,s)} \right).\]
Combining this
  with  Equation~\eqref{e:Gop} gives \eqref{e:linfbound}.
\end{proof}

We return to the proof of Lemma \ref{lem:covering}.  Recall that the goal was to bound
\[ \G(x_L) + \G(x_{\overline{L}}) \lesssim \sqrt{Bs}.\]
By \eqref{eq:spikyflat} and \eqref{e:l1bound}, $\G(x_L) \lesssim \sqrt{Bs}$.  Furthermore,
\begin{align*}
    \G(x_{\overline{L}}) &\lesssim \sqrt{Bs} +  \sqrt{\log m} \left(
      Q_1\sqrt{B} + Q_2\sqrt{\min(k,s)} \right) \sqrt{ s\|x_{\overline{L}}\|_\infty +
      \log k } \\
    &\leq \sqrt{Bs} + \sqrt{\log m} \left( Q_1\sqrt{B} +
      Q_2\sqrt{\min(k,s)} \right) \left(\sqrt{ \frac{s\log m}{k}} + \sqrt{\log k}\right)  \\
&= \sqrt{Bs} \left(1 + Q_1 \left( \frac{\log m}{\sqrt{k}} +
    \sqrt{\frac{\log m\log k}{s}} \right) + Q_2 \left(
    \frac{\sqrt{\min(k,s)}\log m}{\sqrt{kB}} + \sqrt{\frac{\log m\log k}{B}\frac{\min(k,s)}{s}} \right) \right).
\end{align*}
Since we have assumed $B \gtrsim Q_2^2\log^2 m$, the $Q_2$ term is bounded by a constant.  
Further, $k \gtrsim Q_1^2\log^2 m$, and $s \geq B \gtrsim Q_1^2 \log
m\log k$, and so the $Q_1$ term is also constant.
Thus, we conclude
\[
\G(x) \leq \G(x_L) + \G(x_{\overline{L}}) \lesssim \sqrt{Bs},
\]
which was our goal.

\section{Conclusion}
In compressed sensing, it is of interest to obtain RIP matrices $\Phi$ supporting fast (i.e. nearly linear time)
matrix-vector multiplication, with as few rows as possible. 
Not only does fast
multiplication reduce the amount of time it takes to collect measurements, it
also speeds up many iterative recovery algorithms, which are based on repeatedly
 multiplying by $\Phi$ or $\Phi^*$.  Similarly, because of applications
in computational geometry, numerical linear algebra, and others, one wants
to obtain JL matrices with few rows and fast matrix-vector
multiplication.  In this work, we have shown how to construct RIP matrices
supporting fast matrix-vector multiplication, with fewer rows than was
previously known.  Combined with the work of~\cite{KW11}, this also implies
improved constructions of fast JL matrices.   

Our work leaves the obvious open question of removing the two
$O(\log(k\log d))$ factors separating our constructions
from the lower bounds.  It seems that both logarithmic factors come
from the estimation~\eqref{eq:dudleyize}.
It would be interesting to see if they could be removed by more sophisticated chaining 
techniques such as majorizing measures.

\section*{Acknowledgments}
We thank Piotr Indyk for suggesting the construction in this work and
for asking us whether it yields any stronger guarantees for the
restricted isometry property.

\bibliographystyle{abbrv}

\bibliography{allpapers}


\appendix

\section*{Appendix}
\begin{lemma}\LemmaName{operator-inequalities}
  For any complex matrix $A$, $\gennorm{2\rightarrow 2}{A}^2 \le
  \gennorm{1\rightarrow 1}{A} \cdot \gennorm{\infty\rightarrow
    \infty}{A}$.
\end{lemma}
\begin{proof}
First we consider the case of Hermitian $A$, then arbitrary $A$.
   For Hermitian $A$, let $\lambda$ be the largest (in magnitude)
   eigenvalue of $A$ and $v$ be the associated eigenvector.  We have
   \[
   \gennorm{1\rightarrow 1}{A} \geq \frac{\gennorm{1}{Av}}{\gennorm{1}{v}} =
   \frac{\gennorm{1}{\lambda v}}{\gennorm{1}{v}} = |\lambda| =
   \gennorm{2\rightarrow 2}{A}.
   \]

   For arbitrary $A$,
   \[
   \gennorm{2\rightarrow 2}{A}^2 = \gennorm{2\rightarrow 2}{A^*A} \leq
   \gennorm{1\rightarrow 1}{A^*A} \leq \gennorm{1\rightarrow
     1}{A^*}\cdot \gennorm{1\rightarrow 1}{A} =
   \gennorm{\infty\rightarrow \infty}{A}\cdot \gennorm{1\rightarrow
     1}{A}  \]
   as desired. In the last inequality we used the fact that 
   $\|\cdot \|_{\infty\rightarrow\infty}$ is equal to the largest
   $\ell_1$ norm of any row, and  $\|\cdot \|_{1\rightarrow 1}$ is
   equal to the largest $\ell_1$ norm of any column.
\end{proof}

\begin{prop}\label{prop:go-real}
Let $f:\C^d \to \R^{2d}$ act entrywise by replacing $a + b\mathbf{i}$ with $(a,b)$.
For any integer $r$, define $F:\C^{r \times d} \to \R^{2r \times 2d}$ to act entrywise by replacing an entry $a + b\mathbf{i}$ by the $2 \times 2$ matrix
\[\begin{bmatrix} a & -b\\b & a \end{bmatrix}.\]
Recall that $T_k \subset \C^d$ is the set of unit norm $k$-sparse
complex vectors, and let $S_{s} \subset \R^{s}$ be the set of unit norm
$s$-sparse real vectors.  Recall that $\decnorm{\cdot}{X}$ is a norm
on $\C^d$ given by $\decnorm{x}{X} = \max_b \twonorm{ A_b x }$, and
let $\decnorm{\cdot}{\tilde{X}}$ be a norm on $\R^{2d}$ given by
$\decnorm{x}{\tilde{X}} = \max_b \twonorm{ F(A_b) x }$.  Then
\begin{enumerate}
\item \label{stillgood} If \eqref{eq:good} holds, then $\max_b \sup_{x
    \in S_{s}} \twonorm{F(A_b) x} \leq \ell(s)$ for $s \leq 2k$.
\item \label{coveringnumbers} With $\decnorm{\cdot}{\tilde{X}}$ as
  above, we have
\[
\N( T_k, \decnorm{\cdot}{X}, u ) \leq \N(S_{2k}, \decnorm{\cdot}{\tilde{X}}, u ).
\]
\end{enumerate}
\end{prop}
\begin{proof}
  By construction, we have $f(Ax) = F(A)f(x)$, and also
  $\twonorm{f(x)} = \twonorm{x}$.  Further, $f(T_k) \subset S_{2k}$
  and $f^{-1}(S_s) \subset T_s$.  Thus, item \ref{stillgood} follows
  because
  \[ 
  \max_b \sup_{x \in S_s} \twonorm{F(A_b)x} \leq  \max_b \sup_{y \in T_s} \twonorm{F(A_b)f(y)} = \max_b \sup_{y \in T_s} \twonorm{A_by} \leq l(s)
  \]
  Similarly, item \ref{coveringnumbers} follows because for any $x,y
  \in T_k$,
  \begin{align*}
    \decnorm{x - y}{X} &= \max_{b \in [m]} \twonorm{A_b (x - y)} \\
    &= \max_{b \in [m]} \twonorm{ F(A_b) f(x - y) } \\
    &= \max_{b \in [m]} \twonorm{ F(A_b) ( f(x) - f(y))} \\
    &= \decnorm{f(x) - f(y)}{\tilde{X}}. 
  \end{align*}
  Hence
  \[
  N(T_k, \decnorm{\cdot}{X}, u) = N(f(T_k), \decnorm{\cdot}{\tilde{X}}, u) \leq N(S_{2k}, \decnorm{\cdot}{\tilde{X}}, u).
  \]

\end{proof}

\begin{lemma}\label{lemma:largeRV}\LemmaName{large-RV} 
Let $\mathcal{F}$ denote the $d\times d$ Fourier matrix.
  Let $\Omega$ with $\abs{\Omega} = B$ be a random multiset with
  elements in $[d]$, and for $S\subseteq [d]$ let
  $\mathcal{F}_{\Omega\times S}$ denote the $|\Omega| \times |S|$
  matrix whose rows are the rows of $\mathcal{F}$ in $\Omega$,
  restricted to the columns in $S$.
Then for any $t > 1$,
\[ 
	\max_{\abs{S} = k} \enorm{}{\mathcal{F}_{\Omega \times S}} \lesssim
        \sqrt{t(B + k\beta)} 
\]
with probability at least 
\[
	1 - O\left( \exp\left( -\min\left\{ t^2, t\beta \right\}
          \right) \right),
\]
where
\[ \beta = \log^2 k \log d \log B. \]
\end{lemma}
\begin{proof} (Implicit in~\cite{RV08}).
Let $X = \sup_{\abs{S} = k} \enorm{}{ I_k - \frac{1}{B}
  \mathcal{F}_{\Omega \times
    S}^* \mathcal{F}_{\Omega \times S}}$, where $I_k$ is the $k\times
k$ identity matrix.
It is shown in~\cite{RV08} that 
\[
	\E_{\Omega} X \lesssim \sqrt{\frac{k \log^2k \log d \log B}{B}(\E X +
          1)} =: \sqrt{ \frac{ k\beta}{B} (\E X + 1) }.
\]
This implies that 
\begin{equation}\label{eqn:expbound}
	\E X \leq 1 + \frac{ O(k\beta)}{B} =: \alpha.
\end{equation}
Indeed, whenever $x^2 \leq A(x+1)$, we have $x < A + 1$ or else we
conclude $(A+1)^2 \leq A^2 + 2A$.
Let $\alpha$ denote the right hand side of \eqref{eqn:expbound}.  We
may plug this expectation into the proof of Theorem 3.9
in~\cite{RV08}, and we obtain
\[ \Pr\left[ X > C t \alpha \right] \leq 3\exp( - C' t \alpha B/k ) +
2 \exp( - t^2 )\] 
for constants $C$ and $C'$.
In the case $X \le Ct\alpha$, we have
\[ 
	\max_{\abs{S} = k} \enorm{}{\mathcal{F}_{\Omega \times S}}
        \leq \sqrt{
          B(1 + C t \alpha ) } \leq \sqrt{B} + \sqrt{ BC t \alpha},
\]
and so we conclude that
\[ \max_{\abs{S} = k} \enorm{}{\mathcal{F}_{\Omega \times S}} \leq
\sqrt{B} +
O\left( \sqrt{ t(B + k \beta) } \right) \]
with probability at least 
\[ 1 - 3 \exp( - C' t (\beta + B/k) ) - 2\exp(-t^2).\]
\end{proof}

\end{document}